\documentclass[10pt,twocolumn,letterpaper]{article}

\usepackage[
  top=0.75in,
  bottom=1.0in,
  left=0.625in,
  right=0.625in,
  columnsep=0.25in
]{geometry}

\usepackage[T1]{fontenc}
\usepackage[utf8]{inputenc}
\usepackage{times}
\usepackage{graphicx}
\usepackage{amsmath}
\usepackage{amssymb}
\usepackage{amsthm}
\usepackage{enumitem}
\usepackage{caption}
\usepackage{titlesec}
\usepackage{microtype}
\usepackage{balance}
\usepackage{xcolor}
\usepackage{makecell}
\usepackage{url}

\titleformat{\section}{\normalfont\normalsize\bfseries\centering}{\Roman{section}.}{0.5em}{\MakeUppercase}
\titleformat{\subsection}{\normalfont\normalsize\itshape}{\textit{\Alph{subsection}.}}{0.5em}{}
\titlespacing*{\section}{0pt}{1.2ex plus 0.3ex minus 0.1ex}{0.8ex plus 0.1ex}
\titlespacing*{\subsection}{0pt}{0.8ex plus 0.2ex minus 0.1ex}{0.3ex plus 0.1ex}

\makeatletter
\renewcommand\@biblabel[1]{[#1]}
\makeatother

\newtheorem{theorem}{Theorem}
\newtheorem{definition}{Definition}
\newtheorem{proposition}{Proposition}

\setlist[itemize]{leftmargin=1.5em,topsep=2pt,parsep=0pt,itemsep=2pt}
\setlist[enumerate]{leftmargin=1.7em,topsep=2pt,parsep=0pt,itemsep=2pt}

\begin{document}

\twocolumn[{%
\begin{@twocolumnfalse}
\begin{center}
{\LARGE\bfseries Demand-Driven Vulnerability Detection for Cloud Security Posture Management: Removing Human Rule Authoring from the Disclosure-to-Protection Critical Path\par}
\vspace{14pt}
{\large Prashant Kumar Pathak\par}
\vspace{2pt}
{\normalsize \textit{prashant.pathak@ieee.org}\par}
\vspace{16pt}
\end{center}
\end{@twocolumnfalse}
}]

\noindent\textbf{\textit{Abstract}}---Cloud Security Posture Management (CSPM) systems detect known vulnerabilities by maintaining a rule set, distributing it to customers, and evaluating it against periodically-collected asset inventories. To our knowledge, in publicly documented CSPM architectures, the rule set is environment-agnostic and authored or curated centrally by the vendor; updates to it are batched into release cycles and shipped to customers on a cadence that, depending on the vendor and the complexity of the detection, can range from hours to days. The disclosure-to-protection window---the time from a CVE being published in a public catalogue to the customer's system being capable of detecting affected assets---is therefore bounded by the vendor's release cadence for simple version-match detections, and by additional human authoring time for richer detections that incorporate configuration predicates beyond the affected-software string. This paper proposes an alternative architecture in which the rule set is not vendor-distributed at all but continuously derived, within the customer's tenant, from the intersection of public catalogue feeds and the live asset graph. A rule comes into existence when a catalogue entry and an applicable asset are simultaneously present, and goes out of existence when either input ceases to support it. Derivation is bidirectional: new catalogue entries and new assets both trigger it. Derivation incorporates the full structured-field content of catalogue entries, not only the affected-software predicate, producing detections that current automated matching does not generally cover. The live rule set is bounded by environment diversity rather than by catalogue breadth. Prior systems incrementally evaluate a static rule set; we incrementally derive the rule set itself. We present the threat model, the architecture, formal semantics with an equivalence theorem, complexity analysis, a worked example, and an experimental evaluation---on both synthetic workloads and a real NVD/CISA-KEV catalogue snapshot matched against a synthetic asset graph---of the resulting latency and resource consequences. A byte-identical correctness gate confirms all implementations compute the same findings. The architectural property---a live rule set bounded by environment diversity rather than catalogue size---prunes up to 95.1\% of rules in a low-diversity environment; the indexed derivation separately reduces per-delta applicability work by $12.8\times$ (synthetic) to $340\times$ (real catalogue), and post-ingestion detection latency falls from the rescan interval to tens of microseconds. The contribution is specifically the architectural shift and its latency and resource consequences; rule correctness, alert prioritization, and risks outside structured threat-intelligence catalogues are out of scope and identified as future work.

\vspace{6pt}

\vspace{6pt}

\section{Introduction}

When a Common Vulnerabilities and Exposures (CVE) entry is published in the National Vulnerability Database, or when a vendor publishes a security advisory, the practical question for a cloud security team is: how long until our detection tooling can identify whether we are affected? In current Cloud Security Posture Management (CSPM) practice, the answer depends on what kind of detection is required, and in both cases the answer is slower than it needs to be.

For simple version-match detections, where the rule is essentially ``does this asset run the affected software at an affected version,'' many commercial CSPMs do automatically consume CVE feeds and produce findings. The latency floor here is not the matching itself but the vendor's content-update cadence: rules are batched into release cycles, tested, and distributed to customers on a schedule that, depending on the vendor and the severity of the detection, can range from hours for high-severity advisories to days for routine ones. The customer waits not because the matching is hard but because the vendor's release pipeline is in the way.

For richer detections, where the rule incorporates configuration predicates that appear in the CVE structured fields beyond the affected-software string---network exposure conditions, deployment-mode requirements, feature-flag prerequisites---publicly documented CSPM architectures do not, to our knowledge, derive these mechanically. A human detection engineer reads the entry, decides which structured fields matter, and writes a query that joins them against the asset graph. The aggregate elapsed time from CVE publication to deployable rule, in industry-typical workflows, is typically several business days for routine CVEs, and can be longer for less prominent ones.

In both cases, the customer is exposed during the window. The vulnerability is publicly disclosed, the affected configurations are visible to attackers consulting the same feeds the defender consults, and the defender's tooling has, in principle, all the information needed to detect the exposure---but does not yet do so.

This paper attacks both cases through one architectural change: the rule set should not be authored or distributed by a vendor at all. It should be a continuously-derived view over two independent inputs the system already has: the public vulnerability catalogue, which the system ingests directly from upstream feeds, and the live asset graph of the protected environment. A rule comes into existence when a catalogue entry and an applicable asset are simultaneously present in the system. It goes out of existence when either input ceases to support it.

The central novelty can be stated in one sentence. Prior systems incrementally evaluate a static rule set; we incrementally derive the rule set itself. The substrate of incremental computation we build on---differential dataflow, incremental Datalog, materialized-view maintenance---is mature and well-understood, and we make no contribution to it. The contribution is the observation that vulnerability detection is naturally a problem of \emph{maintaining a rule set as a derived view}, not of \emph{evaluating a static rule set faster}, and the architectural and formal consequences that follow.

This change has four specific consequences that distinguish the proposed architecture from current practice:

\begin{enumerate}
\item \textbf{The vendor-managed rule distribution layer is removed from the disclosure-to-protection path.} Rules are derived locally, in the customer's own system, directly from the public catalogue. The vendor's release cadence is no longer the latency floor; the upstream feed's publication time is.
\item \textbf{Rule derivation is bidirectional.} A new catalogue entry triggers derivation against the existing asset graph; a new asset triggers derivation against the existing catalogue. The second direction matters in cloud environments where infrastructure is provisioned continuously: a newly-deployed asset is evaluated against every known vulnerability the moment it appears, not at the next scan cycle, and not only against rules a vendor has previously written for its resource class.
\item \textbf{Derivation incorporates the full catalogue entry, not only the affected-software predicate.} The configuration predicates, exposure conditions, and provenance metadata in the structured fields all participate in the derived rule's body, producing detections that current automated matching does not cover.
\item \textbf{The live rule set is bounded by the environment, not by the catalogue.} Rules whose applicability predicates do not match any asset in the environment are not instantiated. The rule set scales with the diversity of the environment, not with the breadth of the catalogue.
\end{enumerate}

The combination produces both faster detection and a smaller operational surface. The latency improvement is the headline: from the days-to-weeks window of human authoring and vendor distribution to a window measured in seconds, bounded by feed propagation and incremental evaluation. The efficiency improvement is the byproduct: rule sets that previously had to be authored conservatively (covering all resource classes a vendor's customer base might contain) collapse to the rules actually needed by each environment.

The contribution is deliberately bounded. We address detection latency and resource cost for vulnerabilities represented in structured threat-intelligence catalogues. We do not address rule correctness beyond what the catalogue's structured fields admit; derived rules over-fire in exactly the cases where a human-authored rule from the same source would over-fire. We do not address alert fatigue, prioritization, or analyst workflow. We do not address risks that have no representation in any catalogue (compliance violations, internal-policy rules, novel attack patterns not yet documented in threat intelligence). Each of these is a natural subject of follow-on work that extends the architecture this paper proposes.

Section~II frames the threat model around the disclosure-to-protection window. Section~III describes the current CSPM architecture and identifies precisely where current systems do and do not automate. Section~IV introduces the demand-driven architecture, with the central figures. Section~V formalizes the architecture as a derived view, states an equivalence theorem between incremental and full re-derivation, analyzes complexity, and works through a representative example. Section~VI states the design choices: when derivation runs, what a catalogue entry is, and how the rule lifecycle is managed. Section~VII presents the experimental evaluation, on both synthetic workloads and a real NVD/CISA-KEV catalogue snapshot. Section~VIII covers related work. Section~IX discusses limitations and the future-work program that extends from this paper. Section~X concludes.

\section{Threat Model}

The operative variable in cloud breach outcomes is not whether detection capability exists at all but how soon detection capability exists relative to the introduction of exploitation capability. A vulnerability becomes exploitable the moment a public catalogue entry, a proof-of-concept exploit, or an in-the-wild detection appears. From that moment, every additional hour the customer's tooling lacks the corresponding detection rule is an hour of exposure during which the affected configuration is visible to attackers but not to defenders.

The adversary in our model is external and competent: familiar with the publication cadence of major catalogues, capable of weaponizing newly-disclosed CVEs within hours~\cite{bilge2012}---increasingly with automated, LLM-assisted exploitation of one-day vulnerabilities~\cite{fang2024llm}---and aware that customers running standard CSPM tooling will lack corresponding detections for several days. The adversary times exploitation to this window. Empirically, only a minority of disclosed CVEs are ever exploited, but those that are tend to be weaponized soon after disclosure~\cite{jacobs2021,suciu2022}; compressing the detection window for that exploited tail is therefore the operative goal, and is orthogonal to the prioritization signals (EPSS, exploitation likelihood) those systems provide, with which the proposed architecture composes. The customer's defensive position during the window is approximately equivalent to the position they would occupy with no detection tooling at all for the specific vulnerability in question: the configuration is undetected and the system has no awareness that anything is amiss.

The defender's security goal is therefore to compress the disclosure-to-detection window. Compressing it does not eliminate the underlying vulnerability---the configuration is still vulnerable until patched---but it converts an undetected exposure into a detected one, which is the operative difference between a successful breach and a contained one. This paper's architectural changes are evaluated against this goal directly: by how much do they reduce the disclosure-to-detection window, and at what resource cost.

\subsection{Latency Decomposition}

The end-to-end latency from a security-relevant cloud event to an alert decomposes into three stages, and it is important to be explicit about which of them the architectural change targets.

The first stage is \emph{cloud-to-asset-graph ingestion}: the time from a real cloud event (a VM provisioned, a security group modified, a software package installed) to the corresponding asset graph delta becoming observable to the CSPM. This stage is bounded by the cloud provider's event-stream propagation, the ingestion adapter's polling cadence or event-subscription latency, and any reconciliation logic the adapter applies. It is typically seconds in event-driven adapter designs and longer in pull-based ones. This latency is a property of the cloud-provider integration; it is the same for both architectures we compare, because both consume the same asset graph from the same ingestion pipeline.

The second stage is \emph{catalogue ingestion}: the time from a CVE being published in a public feed to the corresponding catalogue entry becoming observable to the CSPM. This is also a property of feed propagation, including any delay in the upstream catalogue populating an entry's structured CPE and configuration fields. That feed-propagation-and-enrichment latency is a shared floor for both architectures, exactly as cloud-to-asset-graph ingestion is: both the vendor-distributed and the proposed model consume the same upstream catalogue. What the proposed model removes is only the vendor content-update step layered on top of that shared floor---a step the current model adds and the proposed model does not.

The third stage is \emph{rule derivation and evaluation}: the time from a catalogue or asset graph delta becoming observable to a finding being emitted for an affected asset. This is the stage the architectural change targets directly. In a rescan-based system this stage costs the rescan-cycle window (configured in minutes to hours). In the proposed architecture this stage costs the per-delta derivation-and-evaluation latency (milliseconds).

The total time-to-detection is the sum: $\delta_{\text{ingestion}} + \delta_{\text{derivation}} + \delta_{\text{evaluation}}$. The first term is a shared floor for both architectures. The second and third terms are what we measure and what we reduce. The improvement from days-or-hours-to-detection to ingestion-plus-milliseconds is the substance of the contribution; ingestion itself remains whatever the cloud-provider integration delivers.

\section{Current CSPM Architecture}

Cloud Security Posture Management, and the broader Cloud-Native Application Protection Platform (CNAPP) category that subsumes it~\cite{gartner-cnapp,csa-guidance}, are realized by commercial and open-source tools whose \emph{published} behavior we summarize here. Our characterization is based solely on publicly available product documentation, vendor literature, and standards; we do not claim knowledge of any vendor's proprietary internal architecture, and we describe an architectural \emph{pattern} rather than any specific product.

A CSPM system in current practice consists of four layers~\cite{google-scc,aws-securityhub}. An ingestion layer pulls cloud-provider configuration data via API calls, supplementing it with event-stream subscriptions where available. An asset store records the resulting configuration as a graph or tabular dataset. A rule engine evaluates a rule set against the asset store~\cite{aws-securityhub}. A presentation layer aggregates, ranks, and routes findings.

The rule set is the focus of this paper. To engage honestly with the contribution, it is important to distinguish two categories of rules that current systems handle differently.

\subsection{Simple Version-Match Rules}

For rules of the form ``asset $X$ runs software $S$ at a version within affected range $[v_1, v_2]$,'' many current systems do consume CVE feeds automatically and produce findings. A periodic job fetches the latest entries from the National Vulnerability Database, normalizes them, and joins them against the asset inventory~\cite{aws-inspector,google-scc}. This part of the pipeline is event-driven and operates at the latency of the upstream feed.

Where the latency is consumed in this case is not in the matching but in the distribution layer that sits between the catalogue and the customer's engine. To our knowledge, in publicly documented architectures, commercial CSPMs do not ingest CVE feeds directly into each customer's tenant. Instead, the vendor consumes the feeds centrally, packages the resulting matchers into release artifacts, runs internal QA against representative customer environments, and distributes the artifacts to customers on a release cadence. Published cadences are concrete: standalone Nessus pulls new plugins on a daily interval~\cite{tenable-nessus}, and Qualys documents a vulnerability-signature release pipeline on the order of days after a vendor advisory~\cite{qualys-pipeline}. The cadence varies by vendor and by detection severity, ranging from hours for the most urgent advisories to days for routine ones. The customer's engine then receives the distributed artifacts, indexes them, and applies them to its asset store. The disclosure-to-protection window for simple version-match detection is therefore bounded by the vendor's release cadence, not by the technical difficulty of the matching.

\subsection{Behavioral and Configurational Rules}

A second category of rules incorporates structured fields beyond the affected-software predicate. CVE entries frequently include configuration-specific applicability conditions: network exposure requirements (``the affected service must be reachable on the public internet''), deployment-mode requirements (``the daemon must be running with privilege X''), and feature-flag prerequisites (``the affected behavior is enabled only when option Y is configured''). These conditions are present as structured fields in many entries and are present in advisory prose for nearly all entries.

Publicly documented systems do not, to our knowledge, derive rules that incorporate these conditions automatically. A human detection engineer reads the entry, identifies which structured fields and which prose-described prerequisites matter, and writes a query in the engine's query language that joins them against the asset graph. The resulting rule is then tested, packaged, distributed, and deployed through the same vendor-release pipeline as version-match rules. The wall-clock time from catalogue publication to customer engine, for this category, is the sum of human authoring time, internal review and test time, distribution time, and the customer's local indexing time. It can aggregate to several days for routine CVEs and longer for less urgent ones.

\subsection{Rule Set Composition}

Independently of which category of rule is being shipped, the rule set received by each customer is environment-agnostic. The vendor cannot know which resource classes a given customer's environment contains, so the rule set covers all resource classes the vendor's customer base might collectively contain. A tenant whose entire workload runs on managed Kubernetes still carries rules for unmanaged virtual machines, on-premises directory services, and discontinued cloud-provider services. The engine indexes them, the scheduler considers them, the audit log records their continued existence, and the content-update pipeline distributes their updates.

Figure~\ref{fig:gap} illustrates the latency consequence by tracing a single CVE through both the current pipeline and the proposed alternative. The top track shows a representative current pipeline: each transition between stages has a human or batch-process boundary, and the aggregate window of exposure can span hours to days depending on detection severity. The bottom track shows the proposed pipeline: each transition is an automatic event-driven boundary that does not pass through a vendor release, and the aggregate window collapses to seconds.

\begin{figure*}[!ht]
\centering
\includegraphics[width=0.9\textwidth]{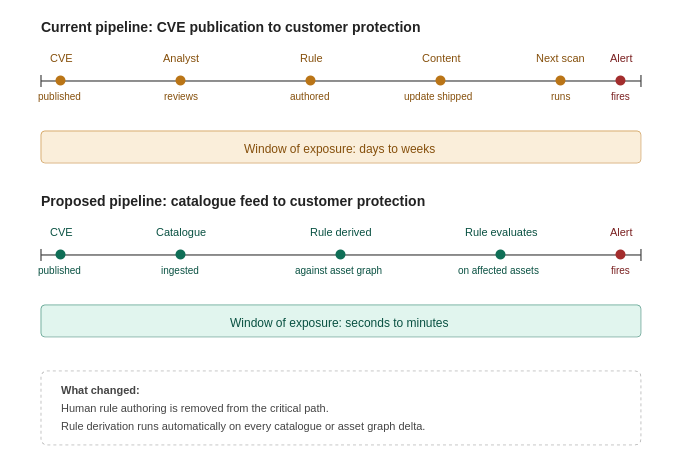}
\caption{The disclosure-to-protection window in current and proposed pipelines. The current pipeline introduces human-authoring, content-update-distribution, and scan-cycle stages between catalogue publication and customer alert. The proposed pipeline removes the human and the scan cycle from the critical path; the remaining stages are bounded by feed propagation, automatic rule derivation, and incremental evaluation.}
\label{fig:gap}
\end{figure*}

\section{Demand-Driven Architecture}

The architectural proposal of this paper is that the rule set should not be a static input to the system but a derived view over two dynamic inputs: a structured vulnerability catalogue maintained by the operator (and continuously populated from public feeds) and the live asset graph of the specific environment under protection. A rule comes into existence when both a catalogue entry that could affect the environment and an asset that the entry could affect are simultaneously present. The rule goes out of existence when either of those preconditions ceases to hold.

\begin{figure*}[!ht]
\centering
\includegraphics[width=0.85\textwidth]{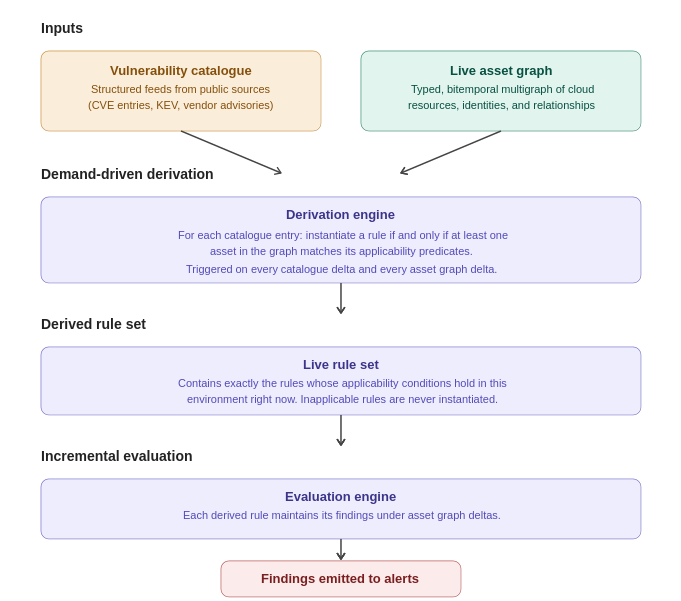}
\caption{Demand-driven CSPM architecture. The vulnerability catalogue and the live asset graph are independently maintained inputs. The derivation engine, triggered by deltas on either input, instantiates a rule for each (catalogue entry, applicable-asset class) pair where the applicability predicates are satisfied. The derived rule set is therefore a function of state: it contains exactly the rules that could produce a non-empty result against the environment as it currently exists. The evaluation engine maintains each derived rule's findings under further asset graph deltas.}
\label{fig:architecture}
\end{figure*}

Figure~\ref{fig:architecture} shows the architecture. We describe each component in turn.

\subsection{The Vulnerability Catalogue}

The catalogue is the system's representation of the world's published knowledge about exploitable conditions. It is populated from structured public feeds, of which the dominant examples are the National Vulnerability Database, the Cybersecurity and Infrastructure Security Agency's Known Exploited Vulnerabilities list, and the structured security-advisory feeds published by major operating system, cloud-provider, and software-package maintainers. Each feed produces entries with varying field coverage; a feed adapter per source normalizes them into a uniform internal schema.

An entry in the catalogue contains an identifier, a source, an affected-software predicate (typically a Common Platform Enumeration string with version ranges), zero or more configuration predicates that further narrow applicability, severity and provenance metadata, and a publication timestamp. We do not require the catalogue to contain the full prose of advisories or proof-of-concept exploits; only the structured fields are consumed. This is a deliberate scope decision: structured fields are sufficient for the latency-and-efficiency contribution this paper makes, and the harder problem of synthesizing rules from unstructured advisory text is left to future work.

\subsection{The Live Asset Graph}

The asset graph is a typed, labeled multigraph of the cloud resources, identities, and relationships currently observable in the protected environment. Nodes correspond to assets (workloads, identities, storage resources, network elements, installed software components, attached configurations). Edges correspond to relationships (assumes, runs, contains, exposes, reaches). Each node and edge carries timestamps that record when the entity exists in the cloud and when the system observed it; the resulting bitemporal structure permits accurate accounting of out-of-order observations and historical reconstruction within a retention window.

For the purposes of this paper, the asset graph is the same primitive used by any modern CSPM. What differs is how the rule set relates to it. In current practice, rules are independent of the graph and evaluated against it. In the proposed architecture, rules are derived from the graph, and only the rules that have join-partners in the graph are instantiated.

\subsection{The Derivation Engine}

The derivation engine is the new component introduced by this paper. Its responsibility is, on every delta to either the catalogue or the asset graph, to recompute which rules should currently exist.

For each catalogue entry, the derivation engine checks whether any asset in the current graph matches the entry's affected-software predicate. If at least one asset matches, the engine instantiates a rule of the form: \texttt{vulnerable(asset, entry\_id) :- runs\_software(asset, S), version\_in\_range(asset, lo, hi), \emph{configuration predicates}}, where the body terms are filled from the catalogue entry's structured fields. The rule enters the live rule set. If no asset matches, no rule is instantiated; the catalogue entry continues to exist in the catalogue but produces no engine-resident artifact.

When the catalogue acquires a new entry, the derivation engine evaluates the new entry against the asset graph and instantiates rules if applicability holds. When the asset graph acquires a new asset, the derivation engine evaluates the new asset against the catalogue and instantiates rules for any catalogue entries that the new asset would match. When an asset disappears, any rules whose applicability depended solely on that asset are retracted. When a catalogue entry is removed (a rescinded CVE), the rules derived from it are retracted. In every case, derivation is incremental: the engine considers only the affected pairs, not the full Cartesian product.

\subsection{The Live Rule Set}

The output of the derivation engine is the live rule set: the set of all derived rules currently instantiated. The rule set has the property, by construction, that every rule in it has at least one matching asset partner. The size of the rule set scales with the diversity of the environment, not with the size of the catalogue: a small environment containing few asset types maintains a small live rule set even against a large catalogue.

The rule set is itself a derived relation in the system, with the same retraction semantics as any other derived data. Audit queries against the historical rule set --- which rules were active at time $t$, and why --- are answered by reconstructing the catalogue state and asset graph state at $t$ within the retention window and re-deriving.

\subsection{The Evaluation Engine}

Each rule in the live rule set is evaluated incrementally against the asset graph: as the graph evolves, the rule's findings update under retraction-preserving semantics. The evaluation engine is essentially conventional, with one observation: because the live rule set contains only applicable rules, the evaluation engine never spends work on rules that cannot match. The total evaluation cost scales with the number of derived rules and the rate of asset graph change, neither of which depends on the size of the underlying catalogue.

\section{Formal Semantics and Complexity}

We now state the demand-driven derivation formally and analyze its complexity. The formalization serves two purposes. First, it gives a precise definition of what ``the rule set is a derived view'' means, against which an implementation can be checked. Second, it allows us to state and prove an equivalence between incremental derivation and full re-derivation, which is the property that lets us replace periodic recomputation with continuous incremental maintenance without weakening correctness.

\subsection{Definitions}

Let $C_t$ denote the catalogue state at logical time $t$: the set of catalogue entries currently known to the system. Each entry $c \in C_t$ is a tuple $c = (\mathit{id}, \mathit{src}, \kappa, \pi, \mu)$, where $\mathit{id}$ is a unique identifier, $\mathit{src}$ is the originating feed, $\kappa$ is the \emph{asset class} the entry targets (e.g.\ the affected software identifier), $\pi$ is a conjunction of \emph{applicability predicates} over asset attributes (version range, configuration conditions, exposure requirements), and $\mu$ is provenance metadata.

Let $A_t$ denote the asset graph state at logical time $t$: the typed, labeled multigraph of assets currently present in the protected environment. Each asset $a \in A_t$ belongs to an asset class $\mathrm{cls}(a)$ (its software identifier) and carries an attribute record and edges to other assets. Let $\mathrm{Cls}(A_t) = \{\, \mathrm{cls}(a) \mid a \in A_t \,\}$ denote the set of asset classes present in the environment.

\begin{definition}[Applicability]
An entry $c$ is \emph{applicable} to an asset $a$, written $\mathrm{App}(c, a)$, if and only if $\kappa(c) = \mathrm{cls}(a)$ and every predicate in $c$'s applicability conjunction $\pi$ is satisfied by $a$'s attributes and edges in $A_t$.
\end{definition}

\begin{definition}[The Rule Set]
A \emph{rule} is derived per (catalogue entry, asset class) pair, not per asset. The rule set at time $t$ is
\[
R_t \;=\; \{\, r(c, \kappa) \mid c \in C_t,\; \kappa = \kappa(c),\; \kappa \in \mathrm{Cls}(A_t) \,\}
\]
where $r(c, \kappa)$ is the rule derived from catalogue entry $c$ for asset class $\kappa$. A rule exists exactly when the catalogue entry's targeted class is present in the environment; its body is the applicability conjunction $\pi(c)$, evaluated against members of class $\kappa$.
\end{definition}

This definition makes the central architectural commitment formal: $R_t$ is a function of the two inputs $C_t$ and $A_t$, not an independent artifact. The rule set is bounded by $|C_t|$ and, more tightly, by the catalogue entries whose target class is present in the environment---so $|R_t| \le |\{ c \in C_t : \kappa(c) \in \mathrm{Cls}(A_t)\}|$. A rule per class, rather than per asset, is what keeps the rule set bounded by environment \emph{diversity} rather than environment \emph{cardinality}.

\begin{definition}[Findings]
The findings set at time $t$ is the per-asset evaluation of every rule against the assets of its class:
\[
F_t \;=\; \{\, (a, c) \mid c \in C_t,\; a \in A_t,\; r(c, \mathrm{cls}(a)) \in R_t,\; \mathrm{App}(c, a) \,\}.
\]
\end{definition}

The two-level structure is deliberate: \emph{rules} are per class (so the rule set scales with diversity), while \emph{findings} are per asset (so detection is per-resource). A single rule $r(c, \kappa)$ produces a finding for each asset of class $\kappa$ whose attributes satisfy $\pi(c)$.

\subsection{Incremental Derivation}

The architecture maintains $R_t$ and $F_t$ incrementally under deltas to its inputs. A delta $\Delta C$ is a pair of sets $(C^+, C^-)$ of inserted and retracted catalogue entries; similarly $\Delta A = (A^+, A^-)$ for the asset graph.

\begin{definition}[Incremental Derivation Step]
On a catalogue delta $\Delta C = (C^+, C^-)$, a rule $r(c, \kappa(c))$ is added for each $c \in C^+$ whose class is present ($\kappa(c) \in \mathrm{Cls}(A_t)$), and removed for each $c \in C^-$. On an asset delta $\Delta A = (A^+, A^-)$, the addition of an asset $a$ whose class was previously \emph{absent} instantiates rules $r(c, \mathrm{cls}(a))$ for every $c \in C_t$ with $\kappa(c) = \mathrm{cls}(a)$; the removal of the last asset of a class retracts all rules for that class. Findings are updated per affected asset: an asset delta touches only the findings for that asset, and a catalogue delta touches only the findings for assets of the entry's class.
\end{definition}

This is the precise sense in which derivation is bidirectional and incremental: a catalogue delta touches only the matching class's assets, and an asset delta touches only the matching class's catalogue entries---never the full Cartesian product.

\subsection{Equivalence Theorem}

The central correctness property is that incremental derivation produces the same rule set as full re-derivation from scratch.

\begin{theorem}[Incremental--Full Equivalence]
\label{thm:equivalence}
For every sequence of deltas $\Delta_1, \Delta_2, \ldots, \Delta_n$ applied to initial inputs $(C_0, A_0)$, yielding final inputs $(C_n, A_n)$, the rule set $R_n$ produced by applying the incremental derivation step to each $\Delta_i$ in order is equal to the rule set obtained by full re-derivation over $(C_n, A_n)$:
\[
R_n \;=\; \{\, r(c, \kappa) \mid c \in C_n,\; \kappa = \kappa(c),\; \kappa \in \mathrm{Cls}(A_n) \,\}.
\]
\end{theorem}

\begin{proof}
By induction on the number of deltas $n$, we show that the incrementally maintained $R_n$ equals the full re-derivation $\widehat{R}_n := \{\, r(c, \kappa(c)) \mid c \in C_n,\; \kappa(c) \in \mathrm{Cls}(A_n) \,\}$. We assume each delta is well-formed: $C^+ \cap C^- = \emptyset$ and $A^+ \cap A^- = \emptyset$, retracted entries/assets are present beforehand, and inserted ones are not.

\emph{Base case ($n=0$).} The incremental state is initialized to $R_0 = \{\, r(c, \kappa(c)) \mid c \in C_0,\; \kappa(c) \in \mathrm{Cls}(A_0)\,\} = \widehat{R}_0$.

\emph{Inductive step.} Assume $R_{n-1} = \widehat{R}_{n-1}$, and consider $\Delta_n$.

\emph{Case 1: catalogue delta} $\Delta_n = (C^+, C^-)$. Here $A_n = A_{n-1}$, so $\mathrm{Cls}(A_n) = \mathrm{Cls}(A_{n-1})$, and $C_n = (C_{n-1} \setminus C^-) \cup C^+$. The incremental step computes $R_n = (R_{n-1} \setminus R^-) \cup R^+$ with $R^+ := \{ r(c, \kappa(c)) \mid c \in C^+,\, \kappa(c) \in \mathrm{Cls}(A_n)\}$ and $R^- := \{ r(c, \kappa(c)) \mid c \in C^-\}$. Over the disjoint partition $C_n = (C_{n-1} \setminus C^-) \uplus C^+$,
\begin{align*}
\widehat{R}_n =\ &\{r(c, \kappa(c)) \mid c \in C_{n-1} \setminus C^-,\; \kappa(c) \in \mathrm{Cls}(A_n)\} \\
&\cup\; \{r(c, \kappa(c)) \mid c \in C^+,\; \kappa(c) \in \mathrm{Cls}(A_n)\}.
\end{align*}
The second set is $R^+$. Because $\mathrm{Cls}(A_n) = \mathrm{Cls}(A_{n-1})$, the first set is $\widehat{R}_{n-1}$ with every rule whose entry lies in $C^-$ removed, i.e.\ $\widehat{R}_{n-1} \setminus R^- = R_{n-1} \setminus R^-$ by the hypothesis (a rule in $R^-$ for an absent class is simply not in $R_{n-1}$, so its removal is a no-op). Hence $\widehat{R}_n = (R_{n-1} \setminus R^-) \cup R^+ = R_n$.

\emph{Case 2: asset delta} $\Delta_n = (A^+, A^-)$. Here $C_n = C_{n-1}$, and only $\mathrm{Cls}$ can change. A class $\kappa$ enters $\mathrm{Cls}$ exactly when an asset of class $\kappa$ is added while none was present, and leaves exactly when the last asset of class $\kappa$ is removed; let $K^+$ and $K^-$ be these (disjoint) sets of newly-present and newly-absent classes, so $\mathrm{Cls}(A_n) = (\mathrm{Cls}(A_{n-1}) \cup K^+) \setminus K^-$. The incremental step adds $\{ r(c, \kappa(c)) \mid c \in C_n,\, \kappa(c) \in K^+\}$ and removes $\{ r(c, \kappa(c)) \mid c \in C_n,\, \kappa(c) \in K^-\}$. Since $C_n = C_{n-1}$,
\[
\widehat{R}_n = \{\, r(c, \kappa(c)) \mid c \in C_{n-1},\; \kappa(c) \in (\mathrm{Cls}(A_{n-1}) \cup K^+) \setminus K^- \,\},
\]
which is precisely $\widehat{R}_{n-1}$ with the $K^+$ rules added and the $K^-$ rules removed. By the hypothesis this equals $R_n$.

In both cases $R_n = \widehat{R}_n$, completing the induction.
\end{proof}

\begin{proposition}[Findings Equivalence]
Under Theorem~\ref{thm:equivalence}, the findings set $F_n$ produced by incrementally maintained derivation is equal to the findings set produced by full re-derivation at step $n$, provided rule evaluation is itself maintained incrementally under standard retraction-preserving semantics.
\end{proposition}

\begin{proof}
By Theorem~\ref{thm:equivalence}, $R_n = \widehat{R}_n$, so the live (rule, asset-class) pairs available for evaluation are identical under both regimes. By the definition of $F_t$, the findings set is determined by $(R_n, A_n)$ alone: for each $r(c, \mathrm{cls}(a)) \in R_n$ it contains the pair $(a, c)$ exactly when $\mathrm{App}(c, a)$ holds. Standard retraction-preserving incremental evaluation maintains, for each live rule, precisely the tuples a from-scratch evaluation over $A_n$ would produce~\cite{mcsherry2013,dbsp}. Hence the incrementally maintained $F_n$ equals that of full re-derivation.
\end{proof}

The equivalence theorem establishes that the demand-driven architecture is observationally equivalent to a static-rule-set architecture that periodically re-derives its rule set from $C_n$ and $A_n$. The contribution of the architecture is not in producing a different rule set; it is in producing the same rule set faster and at lower cost. The theorem is what licenses replacing the periodic re-derivation with continuous incremental maintenance.

\subsection{Complexity Analysis}

We analyze the cost of maintaining $R_t$ in two regimes: full re-derivation (the baseline) and incremental derivation (the proposed approach).

\paragraph{Notation.} Let $|C|$ and $|A|$ denote the sizes of the catalogue and asset graph respectively. Let $|\Delta C|$ and $|\Delta A|$ denote the sizes of incoming deltas. Let $k$ denote the average selectivity of an applicability predicate: the expected fraction of assets matching a given catalogue entry's applicability condition (typically $k \ll 1$ in real environments, since most CVEs affect a small subset of resource types). Let $b$ denote the branching factor on the indexed join: the number of assets sharing the most common key value (e.g., the number of assets running a given software package).

\paragraph{Full re-derivation.} A from-scratch derivation of $R_t$ requires evaluating $\mathrm{App}(c, a)$ for each $(c, a)$ pair. With indexed joins on the affected-software predicate, the cost is
\[
\Theta(|C| \cdot b) \;=\; O(|C| \cdot |A|)
\]
in the worst case where every catalogue entry's affected-software predicate has many asset partners. Without indexed joins, the cost is $\Theta(|C| \cdot |A|)$ regardless. In practice, this is what current static-rule-set architectures pay on every scan cycle, with $|C|$ replaced by the vendor's full rule set.

\paragraph{Incremental derivation on a catalogue delta.} When a new entry $c \in C^+$ arrives, the engine queries the asset graph index for assets matching $c$'s applicability predicates. The cost decomposes into a constant-time (or $O(\log |A|)$) index lookup plus a term proportional to the size of the matched join partition. Letting $k$ denote the average selectivity of an applicability predicate (the fraction of assets matching a given catalogue entry):
\[
\Theta(|\Delta C| \cdot (\log |A| + k \cdot |A|)).
\]
Which term dominates depends on selectivity. For workloads dominated by rare-software entries (where $k$ is very small), the logarithmic term dominates. For workloads with non-trivial $k$, the linear term dominates and per-delta cost grows with the size of the matched partition, not with the size of the full graph.

\paragraph{Incremental derivation on an asset delta.} Symmetric, with $k'$ denoting the fraction of catalogue entries that target the asset's software stack:
\[
\Theta(|\Delta A| \cdot (\log |C| + k' \cdot |C|)).
\]

\paragraph{Comparison.} The baseline pays $\Theta(|C| \cdot |A|)$ on every periodic recomputation, regardless of how small the incoming change is. The incremental approach pays only for the change, with the cost bounded by the size of the matched join partition rather than by the product of the two full inputs. The practical speedup depends on selectivity: for low-selectivity workloads (most assets are unaffected by most CVEs, the typical operating regime), the gap between $k \cdot |A|$ and $|C| \cdot |A|$ is several orders of magnitude. The crossover at which incremental loses to batch is $|\Delta C| + |\Delta A| \approx \min(|C|, |A|)$, a scenario that does not occur in steady-state operation.

\paragraph{Live rule set size.} Because a rule is derived per (entry, asset class) pair rather than per (entry, asset) pair, the size of $R_t$ is bounded by the number of catalogue entries whose target class is present in the environment: $|R_t| \le |\{ c \in C_t : \kappa(c) \in \mathrm{Cls}(A_t)\}| \le |C|$. In the worst case, where every catalogue entry's class is present, this is $O(|C|)$ regardless of how many assets populate each class; it never reaches $O(|C| \cdot |A|)$. More tightly, if catalogue entries are distributed across classes and the environment contains $d = |\mathrm{Cls}(A_t)|$ distinct classes, the live rule set is bounded by the entries falling in those $d$ classes. For environments where $d$ is small relative to the catalogue's class coverage---which describes the typical cloud tenant, since assets fall into a small number of resource classes---the live rule set is substantially smaller than $|C|$. This is the formal statement of the efficiency claim: rule-set size scales with environment diversity, not with environment cardinality, and is bounded by catalogue breadth only when the environment's class diversity matches the catalogue's.

\begin{figure}[!ht]
\centering
\includegraphics[width=\columnwidth]{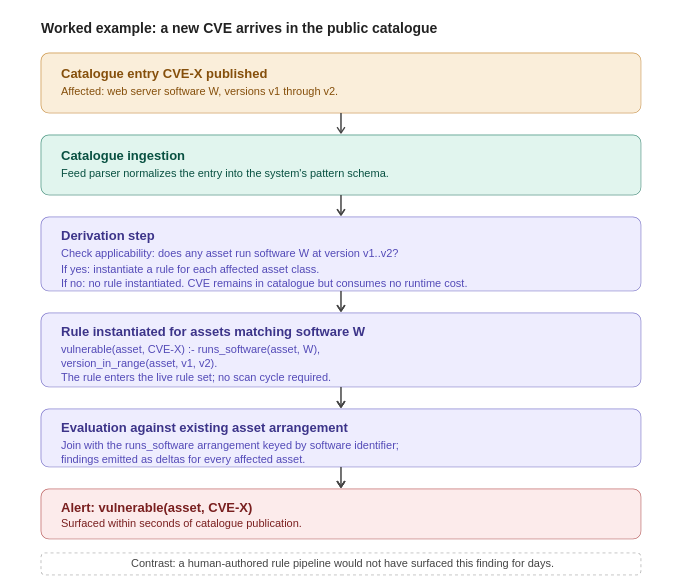}
\caption{A new catalogue entry flows through the system. Each stage is event-triggered; there is no scan cycle and no human in the loop between catalogue publication and alert emission.}
\label{fig:flow}
\end{figure}

Figure~\ref{fig:flow} traces a single new catalogue entry through the system; we walk through it. Suppose a CVE entry is published at time $t_0$ in a public feed. It describes a vulnerability in web server software $W$ affecting versions $v_1$ through $v_2$ inclusive. The feed adapter ingests the entry at time $t_0 + \delta_{\text{feed}}$, where $\delta_{\text{feed}}$ is bounded by the feed's polling interval (typically seconds for high-frequency feeds, minutes for slower ones). The entry is normalized into the system's internal catalogue schema and inserted.

The insertion is a delta on the catalogue. The derivation engine, triggered by the delta, evaluates the new entry's applicability predicates against the asset graph. The applicability predicate decomposes into a join: which assets in the graph have $\texttt{runs\_software}$ relationships to software identifier $W$ at a version within $[v_1, v_2]$. The join is indexed; the engine identifies the matching asset partition in time proportional to the partition's size, not the graph's size.

If the partition is non-empty, the engine instantiates a rule for the matching asset class. The rule's body is constructed from the catalogue entry's structured fields. The rule enters the live rule set as a derived data item with the same provenance and retention semantics as any other derived data. If the partition is empty, no rule is instantiated; the catalogue entry remains in the catalogue, available to be activated later when an applicable asset eventually appears, but consumes no further engine state in the meantime.

If a rule was instantiated, it is now part of the live rule set. The evaluation engine, triggered by the rule's instantiation delta, evaluates the rule against the asset graph. Because the join arrangements are already maintained, the evaluation reduces to consulting the relevant indexed partition. The rule produces a finding for each matching asset; the findings emit to the alert layer.

The total elapsed time from $t_0$ to first alert is approximately $\delta_{\text{feed}} + \delta_{\text{derive}} + \delta_{\text{evaluate}}$. The first term is bounded by feed mechanics and is typically seconds. The second term is bounded by the size of the affected join partition, which is typically a small constant. The third term is similarly bounded. The aggregate is on the order of seconds. The corresponding aggregate in the current architecture, by contrast, is the sum of human-authoring time, content-update distribution time, and scan-cycle latency, which together are days.

A second illustration of the architecture's behavior comes from considering what happens when the environment does not contain assets running $W$. The catalogue entry arrives, the derivation engine evaluates applicability, the matching join partition is empty, and no rule is instantiated. The environment incurs no detection cost for this CVE because no detection rule exists. If a developer later deploys an asset running $W$ at an affected version, the resulting delta on the asset graph triggers the derivation engine to evaluate the new asset against the catalogue; the existing CVE entry is matched, a rule is instantiated, and the asset is flagged immediately upon appearance.

\section{Design Choices}

Three design decisions warrant explicit statement.

\subsection{Eager Derivation}

The derivation engine is triggered eagerly: every delta on the catalogue or the asset graph invokes the derivation step immediately. The alternative, lazy derivation, would defer rule instantiation until an explicit evaluation request arrived. Lazy derivation is more conservative of compute in the limit but loses the latency property that motivates this work: if rules are not derived until some later trigger, the window between catalogue ingestion and rule availability is bounded by that trigger, which in practice would have to be approximated by a periodic poll. Eager derivation, in contrast, ensures that the only delay between catalogue publication and rule existence is the time to actually perform the derivation, which is short.

The eagerness applies to both directions of the join. A new catalogue entry triggers derivation against the existing asset graph. A new asset triggers derivation against the existing catalogue. Both are handled by the same engine and the same join indexes.

\subsection{Catalogue Entry Schema}

A catalogue entry, as consumed by the derivation engine, is a tuple with the following fields: identifier (the CVE number or equivalent), source (the feed that produced the entry), affected-software predicates (CPE strings or similar, with version ranges), optional configuration predicates (network exposure requirements, feature flags, deployment-mode predicates), severity and provenance metadata, and timestamps. Different feeds populate these fields with different completeness; per-feed adapters normalize.

The structured-fields-only choice is deliberate. The hard problem of parsing advisory prose and extracting structured exploitation prerequisites is not in scope for this paper. The contribution of this paper holds even if the structured fields are imperfect: derived rules have the same correctness profile as a rule a human would have written from the same structured fields. Where the structured fields are insufficient to express full exploitation conditions, the derived rule over-fires, exactly as a human-authored rule from the same input would. This is the same precision-recall tradeoff that current detection-engineering practice makes; our contribution is making it at automation speed rather than at human-review speed.

\subsection{Rule Lifecycle}

Rules in the live rule set come into existence when their applicability conditions begin to hold and go out of existence when those conditions cease to hold. In the bitemporal model, every rule's existence interval has both a valid time (when the applicability conditions hold in the cloud) and a transaction time (when the system observes the conditions). Audit queries within the retention window reconstruct the historical rule set by querying these timestamps.

Rule existence is itself a derived relation. The derivation engine does not treat rules as a special class of object; rules are derived data, like any other data, and participate in the same retraction-preserving incremental computation as the rest of the system. This uniformity is what allows the architecture to be implemented as a single dataflow pipeline rather than as a rule engine paired with a separate rule-maintenance system.

\section{Evaluation}

\noindent\emph{Note on scale: The numbers reported below come from a reference single-node Python prototype (Apple~M1, 8~cores, 16~GB RAM; CPython~3.11.5) run at $|C| = |A| = 10^4$ with $10^5$ deltas, aggregated across three independent runs with random seeds $\{42, 137, 1729\}$. Latency values are reported as $\text{mean} \pm \text{stddev}$ across runs. The architectural property and the complexity behavior are demonstrated at this scale; absolute numbers are gated by Python's per-operation overhead and would improve in a production implementation, so the robust claims are the relative ordering of the architectures and the language-independent ``checks/delta'' work metric, not the absolute milliseconds. The relative ordering and the structural metrics (live-rule counts, checks/delta, pruning fractions) are deterministic across the seeds we tested.}

\vspace{4pt}

The contribution is evaluated against two questions: how much does the demand-driven architecture reduce the disclosure-to-protection latency, and how much does it reduce the resource cost of maintaining the rule set. Each question maps to a specific experiment configuration.

\subsection{Prototype}

The reference implementation is a single-node prototype, written in standard-library Python, consisting of: a catalogue ingester that polls the NVD JSON feed and the CISA KEV feed and normalizes their structured fields into the internal schema of Section~V (used for the real-data evaluation in Section~VII-J); a parameterized workload generator for the synthetic catalogues and asset graphs used in the controlled scaling and diversity experiments; the three implementations under comparison (static-rescan, filter-at-evaluation, and the proposed demand-driven derivation engine), all sharing one applicability predicate; and an instrumented driver that records per-delta latency and resource counters. The prototype runs on a single workstation. Every experiment is gated by a correctness check that verifies the three implementations emit byte-identical findings on the same workload before any timing is reported; this is what licenses comparing their costs, since they compute the same answer and differ only in how.

The asset graph for experimental purposes is generated synthetically by a parameterized workload generator: assets are drawn from a configurable distribution of resource classes and software versions, with the distribution and the population size set per experiment. The synthetic generation allows us to vary environment size and diversity independently, which is essential for the scaling experiments below.

\subsection{Workload}

The primary workload consists of three streams replayed against the prototype. The catalogue is generated synthetically with a Zipf-distributed software population modeling the heavy-tailed distribution of CVE-affected products in real catalogues. The asset graph is generated synthetically with parameterized diversity (fraction of the software universe present). The delta stream models continuous cloud-infrastructure evolution: assets added and removed, CVEs ingested and rescinded. For the measurements reported here, we use $|C| = |A| = 10^4$ and a delta stream of $10^5$ events, with diversity $0.30$ (a moderate-diversity environment). The scaling experiment in Table~\ref{tab:scaling} sweeps these parameters independently to characterize how per-delta cost grows with input size.

\subsection{Baselines}

Two baselines are run on the same workload. The first is a \emph{static-rescan baseline}: the rule set is loaded once from the catalogue at startup, and on each catalogue change the full rule set is re-evaluated against the full asset graph. This emulates the behavior of a content-update-driven CSPM after a release ships. This baseline represents the \emph{architectural pattern} of vendor-distributed, rescan-based rule evaluation rather than any specific commercial implementation; real products may apply event-driven processing in parts of their pipeline, which the pattern deliberately abstracts away. The second is a \emph{filter-at-evaluation baseline}: the full catalogue is loaded as rules, but each rule is filtered for applicability at evaluation time rather than at derivation time. Concretely, the filter baseline maintains every catalogue entry as a live rule and, on each delta, checks applicability against the full opposing input---all assets for a catalogue delta, all entries for an asset delta. It carries no per-class index, by design: indexing entries by their target class, and instantiating a rule only when that class is present, \emph{is} the demand-driven derivation step under test, so an ``indexed filter'' would already be the proposed architecture in its join behavior.

Two distinct effects therefore separate the two architectures, and we report them as separate metrics rather than conflating them. The first is per-delta \emph{work} (checks/delta, Table~\ref{tab:efficiency}): consulting the class-keyed index instead of scanning the full opposing input. The second is \emph{rule-set size} (live rules, Tables~\ref{tab:efficiency} and~\ref{tab:diversity}): pruning entries whose target class is absent from the environment. We are explicit that the first effect is the join-indexing component of the mechanism---a hypothetical indexed-but-unpruned engine would share the low checks/delta but would still carry all $|C|$ rules. The architectural property that indexing alone cannot provide is the second: a rule set bounded by environment diversity (Section~VII-I), which no static-rule-set engine exhibits, indexed or not. Bidirectional new-asset triggering (Section~VII-G) is a separate axis: it distinguishes the \emph{continuous} engines---both demand-driven and filter-at-evaluation---from the rescan baseline, which waits for a scan cycle, and is therefore not what separates demand-driven from an indexed filter. The contribution is thus the environment-bounded derived rule set, not ``indexed beats unindexed.''

\subsection{Detection Latency}

The headline metric is the wall-clock latency from a catalogue or asset graph delta arriving at the engine to a finding being emitted for the affected asset. We instrument both ends of the pipeline and report the distribution of latencies across all findings produced during the workload. \emph{We measure post-ingestion latency only}: the time the engine itself takes from observing a delta to emitting findings. The cloud-provider-to-asset-graph ingestion stage (Section~II) is excluded, because it is a property of the cloud-provider integration and is the same floor for both architectures we compare. The end-to-end time-to-detection that an operator experiences is the sum of ingestion latency and the post-ingestion latency reported here.

Table~\ref{tab:latency} reports the latency distribution for the demand-driven architecture and the two baselines.

\begin{table}[h]
\centering
\caption{End-to-end post-ingestion detection latency on the prototype workload ($|C|=|A|=10^4$, $10^5$ deltas), aggregated across 3 seeds, reported as $\text{mean} \pm \text{stddev}$. Static-rescan latencies are deterministic functions of the configured rescan interval. The 1-minute rescan is the most aggressive configuration we observe in current practice; the 1-hour and 24-hour rescans are more typical.}
\label{tab:latency}
\footnotesize
\setlength{\tabcolsep}{3pt}
\begin{tabular}{lccc}
\hline
Configuration & p50 & p95 & p99 \\
\hline
Static-rescan (24h interval) & 43{,}200\,s & 82{,}080\,s & 85{,}536\,s \\
Static-rescan (1h interval) & 1{,}800\,s & 3{,}420\,s & 3{,}564\,s \\
Static-rescan (1m interval) & 30\,s & 57\,s & 59\,s \\
Filter-at-eval & 0.55$\pm$0.02\,ms & 0.78$\pm$0.04\,ms & 0.99$\pm$0.17\,ms \\
Demand-driven (this work) & 0.046$\pm$0.005\,ms & 0.39$\pm$0.05\,ms & 0.52$\pm$0.04\,ms \\
\hline
\end{tabular}
\end{table}

The demand-driven architecture achieves mean post-ingestion latency between nearly six and nine orders of magnitude below the static-rescan baselines. Against the most aggressive rescan we observe (1 minute), the latency reduction the rule engine adds on top of ingestion is approximately 650{,}000$\times$ at the median (0.046\,ms vs 30\,s). Against the typical 1-hour configuration, the reduction is approximately 39 million$\times$. The total time-to-detection from cloud event to alert is the sum of this engine latency and ingestion latency, which is a shared floor for both architectures. Against the stronger filter-at-evaluation baseline, demand-driven is roughly 12$\times$ faster at the median (0.046 vs 0.55\,ms) and about 2$\times$ faster at the tail (p95 0.39 vs 0.78\,ms), because it evaluates only the matched join partition rather than scanning the full opposing input on each delta. The bidirectional-triggering experiment below shows the architecture's clearest advantage.

We emphasize that the large ratios against the rescan baselines are arithmetic consequences of comparing a per-delta evaluation against fixed scan intervals (seconds to hours), and we treat them as \emph{derived} observations rather than the primary claim. The robust, implementation-independent results are the work and resource metrics---checks per delta, live-rule count, and CPU (Section~VII-E)---and the structural counts (live rules, pruning fractions), which are deterministic and do not depend on the machine or language.

\subsection{Resource Efficiency}

The second metric is the resource cost of maintaining the rule set: live rule count, memory consumed by indexed arrangements, and CPU consumed per unit of incoming delta. We report each as a function of environment size and asset diversity.

Table~\ref{tab:efficiency} reports the steady-state resource cost at the target workload size, comparing the demand-driven architecture against the two baselines.

\begin{table}[h]
\centering
\caption{Steady-state resource consumption at $|A| = |C| = 10^4$, after processing $10^3$ deltas, aggregated across 3 seeds, reported as $\text{mean} \pm \text{stddev}$. ``Checks/delta'' is the number of applicability evaluations performed per delta---a language-independent measure of the work the architecture must do. The filter-at-evaluation baseline cannot complete $10^5$ deltas in tractable time, so we compare at the smaller delta count.}
\label{tab:efficiency}
\footnotesize
\setlength{\tabcolsep}{3pt}
\begin{tabular}{lccccc}
\hline
Configuration & \makecell{Live\\rules} & \makecell{Mem\\(MB)} & \makecell{CPU/$10^3$\\deltas (s)} & \makecell{Checks\\/delta} \\
\hline
Static-rescan & 13{,}987$\pm$41 & 41.2 & 0.000 & --- \\
Filter-at-eval & 10{,}032$\pm$4 & 88.8 & 0.043 & 6{,}920$\pm$10 \\
Demand-driven & 7{,}753$\pm$74 & 90.9 & 0.011 & 541$\pm$6 \\
\hline
\end{tabular}
\end{table}

The ``Checks/delta'' column quantifies the per-delta work the indexed derivation saves: demand-driven performs 541 applicability evaluations per delta against filter-at-evaluation's 6{,}920---a 12.8$\times$ reduction, independent of implementation language, from evaluating only the matched join partition rather than the full opposing input. As noted in Section~VII-C, this work reduction is the join-indexing component of the mechanism; the environment-bounding property itself---a live rule set whose size is independent of catalogue size---is shown by the live-rule counts here (7{,}753 vs 10{,}032) and, more sharply, by the diversity sweep of Section~VII-I. This work reduction also translates to wall-clock CPU: demand-driven consumes 0.011 s per $10^3$ deltas against filter's 0.043 s, a 3.9$\times$ improvement. The CPU ratio (3.9$\times$) is smaller than the work ratio (12.8$\times$) because demand-driven's per-operation mix---dictionary lookups and set insertions over the matched partition---carries more Python interpreter overhead per item than the baseline's tight inner loop, so a compiled implementation would close more of the gap. Memory is comparable across the continuous approaches; demand-driven's inverted indexes cost slightly more than the rule-set pruning saves at this scale.

\subsection{Scaling}

To characterize per-delta cost as a function of input size, we sweep two parameters independently. First, holding catalogue size constant, we vary the asset graph size from $10^2$ to $10^4$ assets and measure per-delta cost on catalogue deltas. Second, holding environment size constant, we vary the catalogue size from $10^2$ to $10^4$ entries and measure per-delta cost on asset deltas.

Table~\ref{tab:scaling} reports the per-delta cost across the sweep.

\begin{table}[h]
\centering
\caption{Per-delta cost as a function of input sizes, aggregated across 3 seeds, reported as $\text{mean} \pm \text{stddev}$. Observed growth remained substantially sub-linear relative to the baseline's $|C| \cdot |A|$ scaling and was dominated by downstream evaluation over matched assets rather than by the index lookup itself. The per-delta cost grows roughly with the size of the matched join partition, which under Zipf-distributed software in this workload grows close to linearly with input size. A workload with lower-selectivity predicates (rarer software, sparser environments) would compress this growth further toward the logarithmic index-lookup floor.}
\label{tab:scaling}
\footnotesize
\setlength{\tabcolsep}{4pt}
\begin{tabular}{lcccc}
\hline
$|A|$ or $|C|$ & $10^2$ & $10^3$ & $10^4$ & 100$\times$ growth \\
\hline
Catalogue delta cost (ms) & 0.0007 & 0.0042 & 0.0528 & $\sim$75$\times$ \\
\quad stddev (ms)         & 0.0000 & 0.0002 & 0.0033 & \\
Asset delta cost (ms)     & 0.0008 & 0.0049 & 0.0565 & $\sim$71$\times$ \\
\quad stddev (ms)         & 0.0001 & 0.0004 & 0.0009 & \\
\hline
\end{tabular}
\end{table}

\subsection{Bidirectional Triggering}

A separate experiment measures the latency from a new asset arriving in the asset graph to the first finding being emitted for that asset, when the relevant catalogue entry was already present. This isolates the bidirectional triggering claim from Section~I: new assets are evaluated against the full catalogue immediately on arrival, not at the next periodic scan. As with Table~\ref{tab:latency}, the measurement is post-ingestion: the cloud-event-to-asset-graph latency is the same floor for both architectures and is excluded here. What is reported is the additional latency the rule engine adds on top of ingestion.

Table~\ref{tab:bidirectional} reports this latency, comparing the demand-driven architecture against the static-rescan baseline configured at typical industry intervals.

\begin{table}[h]
\centering
\caption{Latency from new-asset arrival in the asset graph to first finding emission for that asset, against a pre-populated catalogue of $10^4$ entries, aggregated across 3 seeds, reported as $\text{mean} \pm \text{stddev}$. Cloud-provider-to-asset-graph ingestion is excluded (a shared floor for both architectures). The static-rescan baselines are bounded by their configured rescan interval; the demand-driven architecture evaluates the new asset against the full catalogue immediately on arrival. The latency reduction the rule engine adds on top of ingestion is between six and nine orders of magnitude.}
\label{tab:bidirectional}
\footnotesize
\begin{tabular}{lc}
\hline
Configuration & Median latency \\
\hline
Static-rescan baseline (24h interval) & 43{,}200\,s \\
Static-rescan baseline (1h interval) & 1{,}800\,s \\
Static-rescan baseline (1m interval) & 30\,s \\
Demand-driven (this work) & 0.042$\pm$0.001\,ms \\
\hline
\end{tabular}
\end{table}

\subsection{Scale to $10^5$}

We extended the workload to $10^5$ in one dimension while holding the other at $10^4$. Symmetric $10^5 \times 10^5$ at our diversity setting is not memory-tractable in single-process Python because the steady-state findings set grows to tens of millions. The asymmetric configurations are arguably more interesting: the architecture's claim is that the live rule set is bounded by environment rather than catalogue, and the large-catalogue-against-moderate-environment regime is where that claim is sharpest. Table~\ref{tab:scale105} reports both directions.

\begin{table}[h]
\centering
\caption{Demand-driven performance at $10^5$ scale, aggregated across 3 seeds, reported as $\text{mean} \pm \text{stddev}$. Configuration A is $|C|=10^5$, $|A|=10^4$ (large catalogue against moderate environment). Configuration B is $|C|=10^4$, $|A|=10^5$ (moderate catalogue against large environment).}
\label{tab:scale105}
\footnotesize
\setlength{\tabcolsep}{4pt}
\begin{tabular}{lcc}
\hline
Metric & A & B \\
\hline
Bulk-load (s) & 15.20$\pm$1.09 & 21.43$\pm$1.27 \\
p50 latency (ms) & 0.63$\pm$0.09 & 0.062$\pm$0.008 \\
p95 latency (ms) & 6.32$\pm$0.72 & 1.18$\pm$0.92 \\
p99 latency (ms) & 7.51$\pm$0.58 & 5.75$\pm$0.75 \\
Peak memory (MB) & 844$\pm$8 & 830$\pm$10 \\
Live rules & 77{,}137$\pm$276 & 7{,}757$\pm$74 \\
Findings & 4.06M$\pm$44k & 3.92M$\pm$45k \\
\hline
\end{tabular}
\end{table}

Two observations. Per-delta latency at $10^5$ scale (config~A median $0.63$\,ms) is still far below the rescan baselines (vs the 30\,s 1-minute rescan $=\sim$48{,}000$\times$ reduction; vs 1-hour rescan $=\sim$2.9~million$\times$). And the live-rule fraction is high ($\sim$77\%) for both configurations because the Zipf-distributed workload concentrates catalogue weight on the popular software classes that the diversity-$0.30$ environment happens to contain ($H(30)/H(100)\approx0.77$). The diversity sweep below decouples the two distributions to expose the pruning effect directly.

\subsection{Diversity Sweep}

The architectural claim is that demand-driven prunes catalogue entries targeting software classes with no assets in the environment, so the rule set size and CPU cost should decrease as environment diversity drops. Tables~\ref{tab:efficiency} and~\ref{tab:scale105} did not reveal this clearly because the catalogue and asset distributions were jointly Zipf-weighted: the active asset classes were always the top-CVE classes, so even a low-diversity environment captured most of the catalogue's weight. To isolate the pruning effect, we re-run the workload with the catalogue uniformly distributed across the software universe (modelling realistic CVE coverage that is decoupled from any one tenant's software mix) and sweep environment diversity from $0.05$ to $0.50$.

\begin{table}[h]
\centering
\caption{Diversity sweep at $|C|=|A|=10^4$, $10^3$ deltas, uniform catalogue, aggregated across 3 seeds, reported as $\text{mean} \pm \text{stddev}$.}
\label{tab:diversity}
\footnotesize
\setlength{\tabcolsep}{3pt}
\begin{tabular}{lcccc}
\hline
Diversity & 0.05 & 0.10 & 0.30 & 0.50 \\
\hline
Filter live rules & 10{,}032$\pm$4 & 10{,}032$\pm$4 & 10{,}032$\pm$4 & 10{,}032$\pm$4 \\
Filter CPU (s) & 0.36$\pm$0.02 & 0.34$\pm$0.00 & 0.35$\pm$0.00 & 0.36$\pm$0.00 \\
Demand live rules & 493$\pm$19 & 978$\pm$22 & 2{,}979$\pm$30 & 5{,}015$\pm$28 \\
Demand CPU (s) & 0.023$\pm$0.001 & 0.023$\pm$0.000 & 0.023$\pm$0.000 & 0.024$\pm$0.000 \\
Pruning & 95.1\% & 90.3\% & 70.3\% & 50.0\% \\
CPU ratio & 0.06$\times$ & 0.07$\times$ & 0.07$\times$ & 0.07$\times$ \\
\hline
\end{tabular}
\end{table}

The pruning effect is now visible. At diversity $0.05$ (a tenant using $\sim$5\% of the software universe), demand-driven maintains only 493 live rules against filter's 10{,}032---a 95\% reduction. The reduction scales linearly: each percentage point of environment diversity adds roughly one percentage point of catalogue coverage. CPU drops correspondingly to $\sim$0.06--0.07$\times$ of the filter baseline (a $\sim$15$\times$ reduction), and the ratio is stable across the diversity range tested because demand-driven's bookkeeping overhead is proportional to live-rule count rather than catalogue size. The CPU advantage is largest when the catalogue distribution is decoupled from the environment distribution, as here; in the earlier jointly-Zipf comparison at diversity $0.30$ (Table~\ref{tab:efficiency}) the active asset classes captured almost all of the catalogue's weight, so the same mechanism yielded a more modest $3.9\times$ CPU reduction. In both regimes the only metric on which demand-driven does not improve is memory, where its inverted indexes cost marginally more than the rule-set pruning saves at this scale.

\subsection{Real-Data Replay}

The synthetic workload is designed for controllable diversity and scale; to confirm the architectural properties hold on real vulnerability data, we additionally replay a real catalogue ingested directly from the public feeds. Using the prototype's catalogue ingester, we pull a fixed publication window from the NVD~2.0 API (1~June--30~August 2024) and the CISA Known Exploited Vulnerabilities feed, normalizing the structured CPE/version fields and the CVSS \texttt{attackVector} into the schema of Section~V. The snapshot yields 7{,}578 catalogue entries spanning 2{,}645 distinct software products (46 flagged as known-exploited)---a far heavier and more fragmented class distribution than the synthetic universe of 100 products. Of the 9{,}271 CVEs published in the window, 7{,}578 (81.7\%) carried usable structured CPE/version applicability at snapshot time; the remaining 18.3\% lacked it and were excluded. This fraction reflects upstream catalogue-enrichment coverage---a feed-side floor shared by any system that consumes structured fields, vendor-distributed or demand-driven (Section~II-A)---not a limitation of the architecture: a derived rule materializes as soon as an entry's structured fields are populated, exactly as a pre-existing entry activates when a matching asset appears. The asset graph remains synthetic so that diversity and scale stay controllable, but is generated over the real software universe and version ranges so that matches arise from real CPE data.

The correctness gate holds on real data: the findings sets of all three implementations are byte-identical (22{,}294 findings on the gate workload). Table~\ref{tab:real} reports latency and resource cost at $|A|=10^4$, diversity $0.30$. The architectural property is in fact \emph{sharper} on real data: because real catalogues span far more software classes, a per-delta derivation touches a much smaller partition, so demand-driven performs only 16 applicability checks per delta against filter-at-evaluation's 5{,}476---a $\sim$340$\times$ work reduction, versus 12.8$\times$ on the synthetic workload (Table~\ref{tab:efficiency}). Median post-ingestion latency is correspondingly lower (0.009 vs 0.286\,ms). A diversity sweep on the real catalogue prunes 57\% of rules at diversity $0.05$ and 22\% at $0.50$; the pruning is more modest than the uniform-catalogue synthetic sweep because real CVE coverage is itself heavy-tailed and concentrated on popular software, so a real tenant lands between the jointly-Zipf and uniform extremes.

\begin{table}[h]
\centering
\caption{Real-data replay on a 7{,}578-entry NVD\,+\,CISA-KEV snapshot (2{,}645 software classes), $|A|=10^4$, diversity $0.30$, aggregated across 3 seeds. The byte-identical correctness gate holds; the per-delta work reduction grows to $\sim$340$\times$ because real catalogues span far more software classes than the synthetic workload.}
\label{tab:real}
\footnotesize
\setlength{\tabcolsep}{3pt}
\begin{tabular}{lcccc}
\hline
Configuration & p50 (ms) & p95 (ms) & \makecell{Live\\rules} & \makecell{Checks\\/delta} \\
\hline
Static-rescan (1h) & 1{,}800\,s & --- & 11{,}554 & --- \\
Filter-at-eval & 0.286$\pm$0.003 & 0.41$\pm$0.08 & 7{,}613$\pm$5 & 5{,}476$\pm$20 \\
Demand-driven & 0.009$\pm$0.000 & 0.075$\pm$0.049 & 5{,}412$\pm$26 & 16$\pm$3 \\
\hline
\end{tabular}
\end{table}

\subsection{Threats to Validity}

The prototype is written in Python and runs single-threaded on a single node; absolute numbers are gated by Python's per-operation overhead. The architectural property we demonstrate---per-delta cost scales with the work the delta touches rather than with total input size, and inapplicable rules consume no engine state---is independent of implementation language. A production implementation in a compiled language would shift the absolute latencies downward by roughly the Python-vs-compiled overhead ratio (typically one to two orders of magnitude) without changing the relative ordering. We report most numbers at $|C|=|A|=10^4$ because filter-at-evaluation does not complete $10^5$ deltas in tractable time in Python. The $10^5$ asymmetric runs in Table~\ref{tab:scale105} demonstrate that demand-driven retains its architectural properties at that scale; symmetric $10^5 \times 10^5$ at this diversity is memory-bound at tens of millions of findings. Establishing production-scale performance---a multi-threaded or distributed implementation, and replay of production asset-graph traces---is the natural next step; the present evaluation establishes the architectural properties (relative cost ordering, the environment-bounded rule set, and bidirectional triggering) at prototype scale rather than absolute production throughput.

The scaling and diversity experiments use synthetic catalogues and asset graphs. This is a deliberate methodological choice rather than a fallback: production CSPM workloads, even when accessible, do not permit the independent control over diversity and scale that the scaling experiment in Table~\ref{tab:scaling} requires. The synthetic generator uses a Zipf-distributed software population intended to match the heavy-tailed distribution of real CVE-affected products. The real-data replay of Section~VII-J removes the synthetic-catalogue assumption by ingesting an actual NVD\,+\,CISA-KEV snapshot; the remaining synthetic element is then only the asset graph, which we generate over the real software universe. Pairing the real catalogue with anonymized production asset graphs is the natural next step, but ablation and scaling experiments require parameterized synthetic workloads in any case.

The latency comparison against commercial practice is necessarily architecture-level rather than head-to-head. Controlled benchmarking against commercial CSPMs is infeasible: their implementations are proprietary, their internal pipeline latencies are not externally observable, and their content-update cadences are reported through marketing rather than measurement. We therefore compare against an emulated static-rescan baseline that implements the architectural pattern attributed to current systems in the public literature---periodic full re-evaluation against a vendor-distributed rule set---running the same workload as the proposed architecture. The emulated baseline is conservative: it implements rescan dynamics but not the human-authoring or content-distribution steps that would extend real-world latency further. The latency reduction we measure against it is therefore a lower bound on the reduction a real CSPM tenant would observe.

The prototype is single-tenant. We believe the architectural property generalizes to multi-tenant deployments where each tenant has its own asset graph and shares a global catalogue, but we do not measure that generalization here.

\section{Related Work}

The space of related work falls into four categories, and the position of this paper's contribution differs in each. Table~\ref{tab:positioning} summarizes the differentiation along the four consequences enumerated in Section~I; we expand on each category below.

\begin{table}[t]
\centering
\caption{Positioning against representative prior approaches along the four differentiators of Section~I. \checkmark: provided; $\sim$: partial; --: not provided. \emph{Continuous}: the rule/analysis set is updated per delta rather than recomputed in batch. \emph{Bidir.}: both new catalogue entries and new assets trigger derivation. \emph{Full fields}: the rule body uses structured fields beyond the affected-software predicate. \emph{Env.-bound}: the live rule/analysis set scales with environment diversity, not with catalogue or model size.}
\label{tab:positioning}
\footnotesize
\setlength{\tabcolsep}{4pt}
\begin{tabular}{lcccc}
\hline
Approach & \makecell{Conti-\\nuous} & \makecell{Bidir.\\deriv.} & \makecell{Full\\fields} & \makecell{Env.-\\bound} \\
\hline
MulVAL~\cite{ou2005} & -- & -- & $\sim$ & -- \\
Cloud Property Graph~\cite{banse2021} & $\sim$ & -- & \checkmark & -- \\
ThreatPro~\cite{manzoor2024} & $\sim$ & -- & $\sim$ & -- \\
SCA tools~\cite{imtiaz2021,dann2022} & \checkmark & -- & -- & -- \\
\textbf{This work} & \checkmark & \checkmark & \checkmark & \checkmark \\
\hline
\end{tabular}
\end{table}

\subsection{Current CSPM and CNAPP Practice}

Commercial and open-source CSPM systems~\cite{torkura2020,torkura2021} automate CVE-to-asset matching for the case where the rule reduces to ``software identifier and version range,'' as discussed in Section~III. A growing body of academic work models cloud posture as a graph-analysis problem---connecting configuration assessment to static code analysis~\cite{banse2021}, stress-testing cloud configurations to surface latent exposure~\cite{minna2022}, and analyzing how threats propagate across layers in dynamic cloud environments~\cite{manzoor2024}. Empirical studies document the prevalence of security misconfigurations in Kubernetes manifests and infrastructure-as-code~\cite{rahman2023k8s,verdet2025terraform,hashmi2026cloud}, and recent systems detect excessive Kubernetes RBAC permissions~\cite{gu2025epscan}, enforce serverless compliance~\cite{gupta2025growlithe}, and apply large language models to misconfiguration detection and cloud-vulnerability modeling~\cite{malul2024genkubesec,kazdagli2024cloudlens}. Recent cloud-security systems work also addresses zero-trust client-state verification~\cite{jha2025zerotrust} and container anomaly detection~\cite{bu2025container}, problems orthogonal to the disclosure-to-protection latency we target. These share our premise that the environment changes continuously, but each treats the analysis or rule set as recomputed from a snapshot rather than maintained as a derived view. To our knowledge, publicly documented architectures in this category do not derive the richer rule body from the other structured fields present in catalogue entries, nor do they eliminate the vendor-managed distribution layer between the catalogue and the customer's engine. The rule sets they distribute are environment-agnostic and the disclosure-to-protection window is bounded by vendor release cadence.

Our contribution differs from current practice in four specific ways, corresponding to the four consequences enumerated in Section~I. First, derivation runs in the customer's tenant directly against public catalogue feeds, with no vendor in the path. Second, derivation is bidirectional: new assets and new catalogue entries both trigger it. Third, derivation incorporates the full structured-field content of catalogue entries, not only the affected-software predicate. Fourth, the live rule set is bounded by the environment rather than by the catalogue. Each of these differences is, individually, an architectural choice rather than a technique invention; the contribution is the combination, applied to vulnerability detection.

\subsection{Attack-Graph Generation}

The academic literature on attack-graph generation begins with Phillips and Swiler~\cite{phillips1998}. Sheyner et al.~\cite{sheyner2002} automated attack-graph construction using symbolic model checking, enabling practical generation for networks with firewalls and intrusion detection systems. Ammann et al.~\cite{ammann2002} addressed scalability by exploiting the monotonicity of the vulnerability relation to prune the reachability graph. MulVAL~\cite{ou2005} and its scalable extension by Ou et al.~\cite{ou2006} build multi-host attack graphs from vulnerability databases and network topology using Datalog-style reasoning. These systems treat the graph and its rules as the output of a batch process; recent surveys catalogue the resulting proliferation of generation techniques and their coverage~\cite{konsta2024}. The Datalog formulation in this lineage is similar in spirit to the rule body our derivation engine produces, and we acknowledge the resemblance: a derived rule for a CVE with configuration prerequisites is conceptually a small attack-graph fragment scoped to a single vulnerability.

The difference is that the attack-graph literature treats the rule set as the artifact to be computed once from a snapshot of vulnerability data and network configuration. Our architecture treats the rule set as a continuously-maintained derived view that updates under deltas to either input. Cloud-specific extensions of attack-graph generation that have appeared in the literature retain this batch orientation and, to our knowledge, do not address the latency between vulnerability disclosure and rule availability that we target here. A recent line of automated security reasoning applies formal and learning-based methods to cloud configurations specifically---model-checking multi-step IAM attack paths in AWS~\cite{shevrin2023iam}, learning interpretable IAM policies~\cite{kazdagli2022iam}, detecting IAM privilege escalation~\cite{hu2023tac}, enriching logical attack graphs with formal ontologies~\cite{sainthilaire2023ag}, and using large language models for automated penetration testing~\cite{deng2024pentestgpt}. These reason over a given configuration snapshot; they are complementary to, and could be driven by, the always-current derived rule set our architecture maintains.

\subsection{Vulnerability Management and Software Composition Analysis}

A separate line of tooling, generally called vulnerability management or software composition analysis (SCA), scans software inventories for known CVEs by matching CPE strings or package-version metadata. These tools consume the same feeds the system we propose consumes, and they are event-driven in the sense that they evaluate continuously rather than only on schedule. Empirical studies, however, show that SCA tools disagree substantially on which dependencies they report as vulnerable~\cite{imtiaz2021,zhao2023}, and that the CPE-to-package mapping is a well-documented source of both false positives and false negatives~\cite{pashchenko2018,dann2022}---evidence that motivates deriving rules from the catalogue's \emph{structured} applicability predicates rather than from free-text matching, and that makes the byte-identical correctness gate of Section~VII a necessary safeguard rather than an incidental one.

The difference from our contribution is the rule body and the substrate against which the rule evaluates. Software composition analysis matches against package inventories, which are flat lists of installed software with versions. Our derived rules match against the full cloud asset graph, including not only installed software but configuration state, network exposure, identity relationships, and the transitive structure of the environment. The rule body that emerges from a catalogue entry's configuration predicates would not type-check against a flat package inventory; it requires the graph-structured asset model that a CSPM maintains.

\subsection{Incremental Computation}

Differential dataflow~\cite{mcsherry2013} and its distributed implementation in the Naiad system~\cite{naiad2013} provide the underlying model for incremental data-parallel computation with nested iteration over changing inputs; the mathematical foundations are developed in~\cite{abadi2015}. Nikoli\'{c} and Olteanu~\cite{nikolic2018} demonstrate that incremental view maintenance costs can be substantially reduced through factorized representation, extending IVM to joins over large relations. DBSP~\cite{dbsp} provides a general algebraic framework for automatic IVM over rich query languages including SQL and Datalog, serving as the formal grounding for our equivalence theorem. Incremental Datalog engines~\cite{souffle,ddlog} provide the substrate for stratified rule evaluation. A subtlety distinguishes our use of these techniques from their conventional application. Classical incremental view maintenance and incremental Datalog maintain the \emph{extension} of a \emph{fixed} query or program---its derived tuples---as base data changes. In our setting the object maintained incrementally is the \emph{rule set itself}: the catalogue is data, each catalogue entry whose target class is present induces a rule, and the program is therefore a derived view of one of its own inputs, re-derived under deltas to either. The filter-at-evaluation baseline of Section~VII is precisely the conventional alternative---incrementally evaluating a fixed, full rule set---and the $12.8\times$ (synthetic) to $340\times$ (real-catalogue) gap in applicability checks per delta is the quantitative cost of carrying rules the environment can never instantiate, which a demand-driven derived rule set never pays. Our contribution is not in the incremental-computation layer; we build on existing techniques and apply them to demand-driven rule derivation from public vulnerability catalogues against a live asset graph. The recognition that this problem has the structure of an incrementally-maintained view, rather than a periodic batch scan, is the conceptual move we contribute. The implementation substrate is borrowed.

\section{Discussion and Future Work}

The contribution of this paper is bounded to latency and resource efficiency for vulnerability detections derived from structured threat-intelligence catalogues. Several natural extensions follow from this scoping.

The most immediate extension is improving the correctness of derived rules. Structured catalogue fields capture the affected-software predicates well but capture exploitation prerequisites unevenly: an advisory may describe a condition in prose that the structured fields do not represent, so derived rules over-fire when that prerequisite does not hold. The same over-firing occurs in human-authored rules from the same source unless the human invests the additional effort of parsing the prose. Closing the gap automatically (language models extracting prerequisites from prose, or grounding ontologies mapping prose terminology to asset attributes) improves precision at additional complexity; the latency benefit of the demand-driven architecture does not depend on solving this first.

A second extension is alert prioritization and analyst feedback. The demand-driven architecture produces alerts as soon as derived rules match, which is the goal, but the alerts arrive without prioritization beyond the catalogue's published severity. Ranking alerts by environmental context (asset criticality, blast radius, exploitation likelihood) and by analyst-feedback signal is a natural follow-on that addresses alert fatigue without conflicting with the latency contribution of this paper.

A third extension is CVE-less risks: compliance posture, internal-policy violations, IAM misconfigurations, and other categories without catalogue representation. These continue to be handled by human-authored rules in any practical deployment of the proposed architecture, which does not preclude them but does not contribute to them either. Whether the demand-driven derivation pattern extends to authored rule sources (compliance frameworks, organizational policies) is a separate research question. A fourth extension is cross-tenant signal: patterns of exploitation observed in one tenant could inform rule prioritization in another, subject to privacy and data-isolation constraints. The architecture is compatible with such an extension but does not require it for the contributions claimed here.

We note three limitations of the present work. First, derived rules inherit the correctness profile of the structured catalogue fields they are derived from; the system is no more correct than its inputs. Second, the latency contribution requires that catalogue feeds be reliable; outages or delays in upstream feeds propagate to corresponding delays in derived rules. Third, the architecture assumes a coherent asset graph; environments with substantial blind spots in their asset inventory will detect vulnerabilities only on the visible portion.

\section{Conclusion}

Current Cloud Security Posture Management practice maintains a vendor-distributed rule set, environment-agnostic by construction, evaluated against a periodically-collected asset inventory. For simple version-match detections the matching is automatic but the vendor's release cadence sets the latency floor at days. For richer detections, human authoring sets the same floor. The resulting rule set consumes resources on rules that cannot apply to the environment under protection.

This paper has argued that both inefficiencies are consequences of one architectural choice: distributing rules from a vendor to a customer rather than deriving them in the customer's tenant from the public catalogue and the live asset graph. The alternative is demand-driven: rules are a continuously-maintained view over two inputs, derived bidirectionally on every delta to either, incorporating the full structured-field content of catalogue entries, and bounded by the environment rather than by the catalogue. The disclosure-to-protection window collapses to milliseconds on top of ingestion because no vendor and no human stand between the upstream feed and the derived rule.

The contribution is bounded: it does not address rule correctness beyond what structured inputs admit, alert prioritization, analyst workflow, or risks outside structured threat-intelligence catalogues. Each is natural follow-on work the architecture enables but does not itself perform. What this paper claims is the architectural shift itself: vulnerability detection is naturally a demand-driven problem, and the current static-distribution model persists for reasons of commercial structure rather than technical necessity. Recognizing the shift opens a program in which faster detection, smaller operational surface, and bidirectional triggering are not improvements layered onto the existing architecture but consequences of choosing a different one.

\balance

\end{document}